\documentclass[12pt]{amsart}

\usepackage{amscd} 
\usepackage{amsmath,amssymb}
\usepackage{mathtools}
\usepackage{mathrsfs}
\newtheorem{theorem}{Theorem}[section]
\newtheorem{lemma}[theorem]{Lemma}
 
\theoremstyle{definition}

\newtheorem{remark}[theorem]{Remark} 
\numberwithin{equation}{section}

\def\z*{\bar z}

\def\dom{\text{\rm dom}}
\def\ran{\text{\rm ran}}

\def\RE{\mathbb R}
\def\CO{{\mathbb C}}

\def\ph*{\phi_\Diamond}

\def\be{\begin{equation}}
\def\ee{\end{equation}}
\def\min{{\rm min}}

\def\comp{{\rm comp}}

\def\-{{\rm in}}
\def\+{{\rm ex}}

\def\bou{{\mathscr B}}

\begin{document}\title[Inverse wave scattering in the time  domain for point scatterers]
{Inverse wave scattering in the time  domain for point scatterers}
\author{Andrea Mantile}
\author{Andrea Posilicano}

\address{Laboratoire de Math\'ematiques de Reims, UMR9008 CNRS et Universit\'e de Reims Champagne-Ardenne, Moulin de la Housse BP 1039, 51687 Reims, France}
\address{DiSAT, Sezione di Matematica, Universit\`a dell'Insubria, via Valleggio 11, I-22100
Como, Italy}
\email{andrea.mantile@univ-reims.fr}
\email{andrea.posilicano@uninsubria.it}


\begin{abstract} 
Let $\Delta_{\alpha,Y}$ be the bounded from above self-adjoint realization in $L^{2}({\mathbb R}^{3})$ of the Laplacian with $n$ point scatterers  placed at  $Y=\{y_{1},\dots,y_{n}\}\subset{\mathbb R}^{3}$, the parameters $(\alpha_{1},\dots\alpha_{n})\equiv\alpha\in {\mathbb R}^{n}$ being related to the scattering properties of the obstacles. Let $u^{\alpha,Y}_{f_{\epsilon}}$ and $u^{\varnothing}_{f_{\epsilon}}$ denote the solutions of the wave equations corresponding  to $\Delta_{\alpha,Y}$ and to the free Laplacian $\Delta$ respectively, with a source term given by a pulse $f_{\epsilon}$ supported in $\epsilon$-neighborhoods of the points in $X_{N}=\{x_{1},\dots, x_{N}\}$,  $X_{N}\cap Y=\varnothing$. We show that, for any fixed $\lambda>\sup\sigma(\Delta_{\alpha,Y})\ge 0$, there exists $N_{\circ}\ge 1$ such that  the locations of the points in $Y$ can be determined by the knowledge of a finite-dimensional scattering data operator $F^{N}_{\lambda}:{\mathbb R}^{N}\to{\mathbb R}^{N}$, $N\ge N_{\circ}$. Such an operator is defined in terms of the limit as $\epsilon\searrow 0$ of the Laplace transform of $u^{\alpha,Y}_{f_{\epsilon}}(t,x_{k})-u^{\varnothing}_{f_{\epsilon}}(t,x_{k})$, $k=1,\dots, N$.
We exploit the factorized form of the resolvent difference 
$(-\Delta_{\alpha,Y}+\lambda)^{-1}-(-\Delta+\lambda)^{-1}$ and a variation on the finite-dimensional factorization in the MUSIC algorithm. Multiple scattering effects are not neglected; our model can be interpreted as the time-domain version of a frequency-domain scattering from an array of Foldy's point-like obstacles.\end{abstract}
\maketitle
\section{Introduction.}
Given the finite set $Y=\{y_{1},\dots,y_{n}\}\subset \RE^{3}$ and $0<r\ll 1$, let 
$$
\partial_{tt}u=
\Delta_{Y^r}u
$$ 
be the wave equation describing the propagation of acoustic waves in the inhomogeneous medium made of a homogeneous  one containing the array of $n$ small spherical obstacles 
$$
{Y^r}=B^{1}_{r}\cup\dots\cup B_{r}^{n}\,,\qquad B_{r}^{i}=\{x\in\RE^{3}:|x-y_{i}|< r\}\,. 
$$
More precisely, $\Delta_{Y^r}$ is the self-adjoint realization in $L^{2}(\RE^{3})$ of the Laplacian with boundary conditions 
\be\label{BC}
\gamma_{0}^{i}u+\alpha_{i}(r)\, [\gamma^{i}_{1}]u=0\,,\quad i=1,\dots, n\,,
\ee
at the boundaries $S^{i}_{r}=\{x\in\RE^{3}:|x-y_{i}|=r\}$. Here $\gamma^{i}_{0}$ and $[\gamma^{i}_{1}]$,  denote the Dirichlet trace at $S^{i}_{r}$ and the jump across $S^{i}_{r}$ of 
 the Neumann trace $\gamma^{i}_{1}$ respectively; $\alpha_{1}(r), \dots, \alpha_{n}(r)$ are $r$-dependent parameters to be specified later.\par
In our previous work \cite{PAMS}, we considered inverse wave scattering in the time domain for a wide class of self-adjoint Laplacians, including those with hard, soft and semi-transparent bounded obstacles with Lipschitz boundaries.  By applying to $\Delta_{Y^r}$ the results there provided  (which build on our previous works \cite{JDE}, \cite{JST}, \cite{JMPA}, \cite{IP}), one gets the following: denoting by $u^{Y^r}_{f}$ and $u^{\varnothing}_{f}$ the solutions of the wave equations corresponding  to $\Delta_{Y^r}$ and to the free Laplacian $\Delta$ respectively, with a source term $f$ concentrated at time $t=0$ (a pulse) one has that for any fixed $\lambda\ge \lambda_{\circ}> 0$ and any fixed open set $B\subset\subset{\mathbb R}^{n}\backslash\overline Y_{\!r}$, the obstacle $Y^r$ can be reconstructed by the knowledge of the data operator $F^{{Y^r}\!\!,B}_{\lambda}:L^{2}(B)\to L^{2}(B)$, 
\be\label{FB}
F^{{Y^r}\!\!,B}_{\lambda}f:=\int_{0}^{\infty}\!\!e^{-\sqrt\lambda\,t}\,1_{B}\big(u^{Y^r}_{f}(t,\cdot)-u^{\varnothing}_{f}(t,\cdot)\big)1_{B}\,dt\,,\qquad \mbox{supp}(f)\subset B\,.
\ee
Since the choice of the set $B$ where both the source and the detector are placed is arbitrary (beside the constraint $B\cap\overline {Y^r}=\varnothing)$, one is lead to choose $B$ having the same kind of shape as $Y^r$, i.e., $B=X^{\epsilon}$, where $X^{\epsilon}$ denotes the $\epsilon$-neighborhood of a set $X=\{x_{1},\dots,x_{N}\}$  such that $X\cap Y=\varnothing$.
Thus, given  $X$ and $Y$, once the parameters  $\alpha_{i}(r)$ and $\lambda$ have been fixed, the data operator $F_{\lambda}^{{Y^r}\!\!,X^{\epsilon}}$ in \eqref{FB} depends on $r$ and $\epsilon$ alone and a natural question arises: what happens whenever $r\searrow 0$ and $\epsilon\searrow 0$ ? In more detail:
\vskip6pt
1) is there a well defined limit self-adjoint operator $\Delta_{\alpha,Y}$ describing the propagation of acoustic waves in an otherwise homogeneous medium containing an array $Y$ of point scatterers ?
\vskip6pt
2) is such an array  $Y$ determined by a finite-dimensional scattering data operator 
$F^{X}_{\lambda}$ corresponding to $X$ and to the wave dynamics generated by $\Delta_{\alpha,Y}$ ? 
\vskip6pt\noindent 
The answer to the first question has been known from a long time:  by \cite[Theorem 2]{FT} (see also \cite[Lemma 2.2]{AGS} and \cite[Theorems 3.4 and 3.7]{Shi} for the case of a single sphere),  setting 
\be\label{scaling}
\alpha_{i}(r)=r+4\pi\alpha_{i}\,r^{2}+o(r^{2})\,,\qquad \alpha\equiv(\alpha_{1},\dots\alpha_{n})\in \RE^{n}\,,\ee 
in \eqref{BC},  $\Delta_{Y^r}$ converges as $r\searrow 0$ (in strong resolvent sense)  to a well defined self-adjoint, bounded from above, operator $\Delta_{\alpha,Y}$. Such an operator $\Delta_{\alpha,Y}$ was firstly rigorously defined in the seminal paper \cite{BF} as a self-adjoint extension of the Laplacian restricted to smooth functions with compact support disjoint from $Y$.  Since then it attracted an increasing attention and has been used in a wide range of applications: we refer to the huge list of references in [1], the main text devoted to this operator and its ramifications. In next Section 2 we will recall the definition of $\Delta_{\alpha,Y}$ and describe its main properties. \par
Although the origin of $\Delta_{\alpha,Y}$ has its root in Quantum Mechanics, as well as most of its applications (however see \cite{NP98} for its connections with electrodynamics of point particles), in recents years it has been used to provide a rigorous mathematical framework for Foldy's scattering of time-harmonic acoustic and elastic waves, see \cite{HMS} and \cite{HMST}.  
While Foldy's approach considers wave scattering  in the frequency domain  (see \cite{Foldy}, \cite[Section 8.3]{Martin}), our aim here is to 
work in the time domain. 
\par
Let us point out that the scaling \eqref{scaling} entering into the boundary conditions \eqref{BC} is the only one leading to a not trivial limit dynamic, i.e., a dynamic different from propagation of free waves in the whole space. This is reminiscent of the case in which one approximates points scatterers with a scaled potential, where the limit dynamics is not trivial if and only if the unscaled potential has a zero-energy resonance (see \cite[Section I.I.2]{AGHKH}). There is an analogous phenomenon whenever one approximate a point scatterer with an obstacle modeled by a shrinking sphere: the boundary conditions \eqref{BC} together with \eqref{scaling} provide a zero-energy resonance (see \cite[Theorems 3.7]{Shi}); different boundary conditions lead, in the limit $r\searrow 0$, to the free Laplacian. For example, whenever one consider Dirichlet boundary conditions on an array of shrinking spheres, one gets an expansion (w.r.t. the radius $r\ll 1$) of the scattered waves which contains no zero-order term (see \cite{SWY}; notice that the coefficients $C_{j}$ appearing there in the expansion (1.7)  are proportional to the radii of the shrinking spheres). 
\par
Taking into account  the limit $r\searrow 0$, one has then at disposal the well defined data operator $F_{\lambda}^{X^{\epsilon}}:L^{2}(X^{\epsilon})\to L^{2}(X^{\epsilon})$,
$$
F_{\lambda}^{X^{\epsilon}}f:=\lim_{r\searrow 0}F_{\lambda}^{{Y^r}\!\!,X^{\epsilon}}f=\int_{0}^{\infty}e^{-\sqrt\lambda\,t}\,1_{X^{\epsilon}}\big(u^{\alpha,Y}_{f}(t,\cdot)-u^{\varnothing}_{f}(t,\cdot)\big)1_{X^{\epsilon}}\,dt\,,\qquad \mbox{supp}(f)\subset X^{\epsilon}\,.
$$
where $u^{\alpha,Y}_{f}$ denotes the solutions of the wave equation corresponding  to $\Delta_{\alpha,Y}$ with source term given by the pulse $f$ (see Remark \ref{WD}).\par
As we will recall in Subsection \ref{MS} below, a relevant point  is the fact that the limit wave equation generated by $\Delta_{\alpha,Y}$, i.e., 
\be\label{we}
\partial_{tt}u=
\Delta_{\alpha,Y}u
\ee 
 can be recasted into the distributional form 
\be\label{es}
\partial_{tt} u=\Delta u+\sum_{i=1}^{n}q_{i}(t)\, \delta_{y_{i}}\,.
\ee
Here $\delta_{y_{i}}$ denotes Dirac's delta distribution at $y_{i}$ and the $q_{i}(t)$'s evolve according to a first-order retarded differential equation (see \cite{KP}, \cite{NP05} and \eqref{ret} below); the terms containing the retardation provide the contributions due to multiple scattering. \par 
Let us point out that, for a time-harmonic solution $u(t,x)=e^{-it\kappa}u_{\kappa}(x)$, the equation \eqref{we} recasts into the generalized eigenfunctions problem
$$
(\Delta_{\alpha,Y}+\kappa^{2})u_\kappa=0\,,
$$ 
whose solutions are known to be of the kind (see, e.g., \cite[eq. (1.5.1), Section II.1.5]{AGHKH})
\be\label{sergio}
u_\kappa(x)=u_{\kappa}^{in}(x)+\frac1{4\pi}\sum_{1\le j,\ell\le n}\Lambda^{j\ell}_{-\kappa^{2}+i0}\,u_{\kappa}^{in}(y_{\ell})\,\frac{e^{i\kappa |x-y_{j}|}}{|x-y_{j}|}\,,
\ee
where $u_{\kappa}^{in}$ solves the free Helmholtz equation, while $\Lambda^{j\ell}_{-\kappa^{2}+i0}:=\lim_{\varepsilon\searrow 0} \Lambda^{j\ell}_{-\kappa^{2}+i\varepsilon}$ and $\Lambda_{\zeta}$ is the inverse of the matrix \eqref{matrix M} defined in the next Section. Since 
$$
\Lambda_{-\kappa^{2}+i0}=\text{Diag}(\alpha_{j}-i\kappa/4\pi)A_{\kappa}\,,
$$ 
where the matrix  $A_{\kappa}$ is defined as in \cite[page 7, Section 2.3.1]{CS}, formula \eqref{sergio} describes the scattering by an array $Y$ of Foldy's point scatterers having scatterings coefficients 
$$
g_{j}=\frac1{\alpha_{j}-\frac{i\kappa}{4\pi}}
$$ 
(see \cite[eqs. (2.9) and (2.19)]{CS}). \par 
The aim of the present paper is to give a positive answer to the second question. We show that the limit data operator $F^{X}_{\lambda}:=\lim_{\epsilon\searrow 0} F^{X^{\epsilon}}_{\lambda}$ is a well defined map on $\RE^{N}$ to itself (see Lemma \ref{LF} for the complete result). Moreover, denoting by 
$P^{X}_{\lambda}$ the orthogonal projector onto 
$\ker(F^{X}_{\lambda})$ and by $\phi_{\lambda}^{X}(z)\in \RE^{N}$  the vector with components 
$$
\big(\phi_{\lambda}^{X}(z)\big)_{k}=\frac{e^{-\sqrt\lambda\,|x_{k}-z|}}{|x_{k}-z|}\,, \qquad z\in \RE^{3}\backslash X\,,\quad x_{k}\in X\,,
$$ 
we show in Theorem \ref{thm} (to which we refer for the precise statement) that the set $Y$ is determined according to the relation
$$Y=\big\{\text{peak points of the function $\RE^{3}\backslash X\ni z\mapsto |P_{\lambda}^{X}\phi^{X}_{\lambda}(z)|^{-1}$}\big\}\,.
$$ 
Our results can be read as a time-domain analogue of the inverse scattering by point-like scatterers in the Foldy regime studied in \cite[Section 2.3.1]{CS}; they provide the counterpart, in the case of point scatterers, of our previous results (see \cite{PAMS}) on time-domain inverse scattering for extended obstacles.\par  
The main ingredients in our proofs are the factorized form of the resolvent difference $(-\Delta_{\alpha,Y}+\lambda)^{-1}-(-\Delta+\lambda)^{-1}$ and a variation on the factorization method approach to the MUSIC (MUltiple-SIgnal-Classification) algorithm provided by Kirsch in \cite[Section 2]{K2002} (see also \cite[Section 4.1]{KG}); however here, contrarily to the frequency-domain case treated by Kirsch, the multiple scattering effects are not neglected. 

\section{Laplacians with point scatterers.}
Given $Y=\{y_{1},\dots, y_{n}\}\subset\RE^{3}$, to any $\alpha\equiv(\alpha_{1},\dots\alpha_{n})\in \RE^{n}$ there corresponds  the  self-adjoint realization in $L^{2}(\RE^{3})$ of the Laplacian with $n$ point scatterers $y_{1},\dots, y_{n}$ defined by
\begin{align*}
\dom(\Delta_{\alpha,Y})=\bigg\{&u\in L^{2}(\RE^{3}):u(x)=u_{0}(x)+\frac1{4\pi}\sum_{j=1}^{n}\frac{\xi_{j}}{|x-y_{j}|}\,,\ u_{0}\in \dot H^{2}(\RE^{3})\,,\\ 
&\xi\equiv(\xi_{1},\dots,\xi_{n})\in\CO^{n}\,,\ \lim_{x\to y_{j}}\left(u(x)-\frac1{4\pi}\,\frac{\xi_{j}}{|x-y_{j}|}\right)=\alpha_{j}\xi_{j}\bigg\}\,,
\end{align*}
$$
\Delta_{\alpha,Y}:\dom(\Delta_{\alpha,Y})\subset L^{2}(\RE^{3})\to L^{2}(\RE^{3})\,,\qquad
\Delta_{\alpha,Y}u:=\Delta u_{0}.
$$
We refer to \cite[Chapter II.1]{AGHKH} for more details and proofs. Here the homogeneous Sobolev space of order two $\dot H^{2}(\RE^{3})$ is defined by 
$$
\dot H^{2}(\RE^{3}):=\big\{u\in{\mathscr C}_{b}(\RE^{3}):|\nabla u|\in L^{2}(\RE^{3})\,,\ \Delta u\in L^{2}(\RE^{3})\big\}
$$
and its relation with the usual Sobolev space of order two $H^{2}(\RE^{3})$ is given by $H^{2}(\RE^{3})=\dot H^{2}(\RE^{3})\cap L^{2}(\RE^{3})$.\par
The operator $\Delta_{\alpha,Y}$ belongs to the set of self-adjoint extensions of the symmetric one 
$S_{Y}$ given by the restriction of the free Laplacian to functions vanishing at the points in $Y$, i.e., 
$S_{Y}:=\Delta|{\mathscr C}^{\infty}_{comp}(\RE^{3}\backslash Y)$; the vector $\alpha\in \RE^{n}$ plays the role of the extension parameter. \par
The extensions of the kind $\Delta_{\alpha,Y}$ suffice for the description of the relevant physical models: by \cite[Theorem 4]{KP}, the wave equation $\partial_{tt}u=Au$ corresponding to a self-adjoint extension $S_{Y}\subset A\subset S_{Y}^{*}$ has a finite speed of propagation if and only if $A=\Delta_{\alpha,Y}$ for some $\alpha\in \RE^{n}$. Moreover, finite speed of propagation holds if and only if the boundary conditions at $Y$ specifying the self-adjointness domain are of local type, i.e., they do not couple scatterers placed at different points: the scatterers are independent of each other.      \par
The vector $\alpha\in \RE^{n}$, beside specifying  the boundary conditions at $Y$, 
is related to the scattering length $a$ of the scatterers through the relation $a=-{(4\pi)^{-1}}\sum_{i=1}^{n}\alpha^{-1}_{i}$ (see \cite[Section II.1.5]{AGHKH}).\par 
The resolvent of $\Delta_{\alpha,Y}$ is  given by 
\be\label{krein}
(-\Delta_{\alpha,Y}+\zeta)^{-1}=(-\Delta+\zeta)^{-1}+K_{\zeta}\,,\qquad \qquad \zeta\not\in\sigma(\Delta_{\alpha,Y})\,,
\ee
where $\sigma(\Delta_{\alpha,Y})$ denotes the spectrum of $\Delta_{\alpha,Y}$, 
$$
(-\Delta+\zeta)^{-1}:L^{2}(\RE^{3})\to H^{2}(\RE^{3})\,,\qquad \zeta\in\CO\backslash(-\infty,0]\,,
$$
is the resolvent of the free Laplacian with kernel function
$$
(-\Delta+\zeta)^{-1}(x,y)=\frac1{4\pi}\,\frac{e^{-\sqrt \zeta\,|x-y|}}{|x-y|}\,,\qquad\text{Re}(\sqrt \zeta\,)>0\,,
$$
and the finite-rank operator $K_{\zeta}:L^{2}(\RE^{3})\to L^{2}(\RE^{3})$ has kernel function 
$$
K_{\zeta}(x,y)=\frac1{(4\pi)^{2}}\sum_{1\le i,j\le n} \Lambda_{\zeta}^{ij}\,\frac{e^{-\sqrt \zeta\,|x-y_{i}|}}{|x-y_{i}|}\,
\frac{e^{-\sqrt \zeta\,|y-y_{j}|}}{|y-y_{j}|}\,.
$$
Here $\Lambda_{\zeta}\equiv(\Lambda^{ij}_{\zeta})$ is the inverse of the $n\times n$ matrix $M_{\zeta}\equiv(M^{ij}_{\zeta})$
given by 
\be\label{matrix M}
M^{ij}_{\zeta}=\left(\alpha_{i}+\frac{\sqrt \zeta}{4\pi}\,\right)\delta_{ij}-
\frac{e^{-\sqrt \zeta\,|y_{i}-y_{j}|}}{4\pi\,|y_{i}-y_{j}|}\ (1-\delta_{ij})\,,
\ee
$\delta_{ij}$ denoting Kronecker's delta.
Regarding the spectral profile, since the resolvent of $\Delta_{\alpha,Y}$ is a $n$-rank perturbation of the free resolvent, the essential spectrum of the free Laplacian is preserved and the discrete spectrum contains at most $n$ distinct eigenvalues; in more detail 
$$
\sigma_{\text{ac}}(\Delta_{\alpha,Y})=\sigma_{\text{ess}}(\Delta_{\alpha,Y})=(-\infty,0]\,,\qquad 
\sigma_{\text{disc}}(\Delta_{\alpha,Y})=\{\lambda>0:\det M_{\lambda}=0\}\,.
$$
For later purposes, we need to investigate the positiveness of the matrix in \eqref{matrix M}:
\begin{lemma}\label{form} Let the matrix $M_{\lambda}\equiv(M^{ij}_{\lambda})$ be defined by \eqref{matrix M} with $\zeta=\lambda\in (0,+\infty)\,$. Then,
$$\text{$M_{\lambda}$ is positive-definite}\quad\iff\quad\lambda>\lambda_{\alpha,Y}:=\sup\sigma(\Delta_{\alpha,Y})\,.
 $$
\end{lemma}
\begin{proof} Let  $\lambda>0$ and $$
v_{\lambda}^{\xi}(x):=\frac1{4\pi}\,\sum_{j=1}^{n}\xi_{j}\,\frac{e^{-\sqrt \lambda\,|x-y_j|}}{|x-y_j|}\,,\quad \xi\equiv(\xi_{1},\dots,\xi_{n})\in\CO^{n}\,.
$$
By \cite[Section 2]{teta}, the quadratic form $Q_{\alpha,Y}$ of $-\Delta_{\alpha,Y}$ has the $\lambda$-independent representation 
\begin{align*}
\dom(Q_{\alpha,Y})=&\{u\in L^{2}(\RE^{3}):u=u_{\lambda}+v_{\lambda}^{\xi}
\,,\ u_{\lambda}\in H^{1}(\RE^{3})\,,\ \xi\in\CO^{n}\}\,,
\\
Q_{\alpha,Y}(u)=&\|\nabla u_{\lambda}\|^{2}_{L^{2}}+\lambda\,\| u_{\lambda}\|^{2}_{L^{2}}-\lambda\,\| u\|^{2}_{L^{2}}+\langle\xi,M_{\lambda}\xi\rangle\,.
\end{align*}
If $\lambda>\lambda_{\alpha,Y}$ then, for any $\xi\not=0$, 
$$
\langle\xi,M_{\lambda}\xi\rangle=Q_{\alpha,Y}(v_{\lambda}^{\xi})+\lambda\,\|v_{\lambda}^{\xi}\|^{2}_{L^{2}}>0\,.
$$
Conversely, let $\lambda>0$ be such that $M_{\lambda}$ is positive-definite and, given any $u\in \dom(Q_{\alpha,Y})\backslash\{0\}$, let use the decomposition $u=u_{\lambda}+v_{\lambda}^{\xi}$. Then
\begin{align}\label{QF}
Q_{\alpha,Y}(u)+\lambda\,\| u\|^{2}_{L^{2}}=&\,\|\nabla u_{\lambda}\|^{2}_{L^{2}}+\lambda\,\| u_{\lambda}\|^{2}_{L^{2}}+\langle\xi,M_{\lambda}\xi\rangle\nonumber\\
\ge&\begin{cases} \langle\xi,M_{\lambda}\xi\rangle\,,&\xi\not=0\\
\|\nabla u_{\lambda}\|^{2}_{L^{2}}+\lambda\,\| u_{\lambda}\|^{2}_{L^{2}}\,,&\xi=0\end{cases}\\
>&\,0\nonumber
\end{align}
and so $\lambda>\lambda_{\alpha,Y}$.
\end{proof}
Obviously, whenever $Y$ is the singleton $Y=\{y\}$ one has
$$
\lambda_{\alpha,y}=\begin{cases}0&\alpha\ge 0\\
(4\pi\alpha)^{2}&\alpha<0\,.\end{cases}$$
The next result provides a simple rough estimate on $\lambda_{\alpha,Y}$ whenever $n>1$.
\begin{lemma}\label{coerc} Set
$$
\alpha_{\circ}:=\underset{1\le i\le n}\min\alpha_{i}\,,
\qquad d:=\underset{i\not=j}\min\,|y_{i}-y_{j}|\,.
$$ 
Then
$$ 
0\le\lambda_{\alpha,Y}\le \begin{cases}
0\,,&4\pi \alpha_{\circ}d\ge n-1\,,\\
\lambda_{\circ}\,,&4\pi \alpha_{\circ}d< n-1\,,\end{cases}
$$
where  $\lambda_{\circ}>0$ solves
$$4\pi\alpha_{\circ} d+\sqrt{\lambda_{\circ}}\,d=(n-1)\, e^{-\sqrt{\lambda_{\circ}}\,d}\,.
$$ 
\end{lemma}
\begin{proof} The thesis is consequence of \eqref{QF} and the inequality
\begin{align*}
\langle\xi,M_{\lambda}\xi\rangle=&\sum_{j=1}^{n}\left(\alpha_{j}+\frac{\sqrt\lambda}{4\pi}\right)|\xi_{j}|^{2}-\sum_{j<k}
\frac{e^{-\sqrt\lambda\,|y_{j}-y_{k}|}}{4\pi\,|y_{j}-y_{k}|}\ 2\,\text{Re}(\bar\xi_{j}\xi_{k})\\
\ge & \left(\alpha_\circ+\frac{\sqrt\lambda}{4\pi}\right) |\xi|^{2}-\frac{e^{-\sqrt\lambda\,d}}{4\pi\,d}\sum_{j<k}2\,|\xi_{j}|\,|\xi_{k}|\\
\ge & \left(\alpha_\circ+\frac{\sqrt\lambda}{4\pi}-(n-1)\ \frac{e^{-\sqrt\lambda\,d}}{4\pi\,d}\right) |\xi|^{2}\,.
\end{align*}
\end{proof}
\section{Wave scattering and the data operator.}\label{S3} 
\subsection{Abstract wave equations.} Let  $A:\dom(A)\subset L^{2}(\RE^{3})\to L^{2}(\RE^{3})$ be self-adjoint and bounded from above; we consider the Cauchy problem for the corresponding wave equation
\be
\begin{cases}\label{Cpb}
\partial_{tt} u(t)=Au(t)\\
u(0)=u_{0}\in L^{2}(\RE^{3})\\
\partial_{t} u(0)=v_{0}\in L^{2}(\RE^{3})\,.
\end{cases}
\ee
We say that $u\in {\mathscr C}(\RE_{+};L^{2}(\RE^{3}))$ is a mild solution of \eqref{Cpb} whenever, 
for any $t\ge 0$, there holds
$$
\int_{0}^{t}(t-s)u(s)\,ds\,\in\dom(A)\quad\text{and}\quad
u(t)=u_{0}+tv_{0}+A\int_{0}^{t}(t-s)u(s)\,ds\,.
$$
By \cite[Proposition 3.14.4, Corollary 3.14.8 and Example 3.14.16]{A}, the unique mild solution of \eqref{Cpb} is given by 
\be\label{sol}
u(t)=\mbox{Cos}_{A}(t)\,u_{0}+\mbox{Sin}_{A}(t)\,v_{0}
\ee
where the $\bou(L^{2}(\RE^{3}))$-valued  functions $t\mapsto\mbox{Cos}_{A}(t)$ and $t\mapsto\mbox{Sin}_{A}(t)$ are defined through the $\bou(L^{2}(\RE^{3}))$-valued (inverse) Laplace transform by the relations  
\be\label{cos}
\sqrt\lambda\,(-A+\lambda)^{-1}=\int_{0}^{\infty}e^{-\sqrt\lambda\, t}\,\mbox{Cos}_{A}(t)\,dt\,,\qquad \lambda>\lambda_{A}\,,
\ee
\be\label{sin}
(-A+\lambda)^{-1}=\int_{0}^{\infty}e^{-\sqrt\lambda\, t}\,\mbox{Sin}_{A}(t)\,dt\,,
\qquad\lambda>\lambda_{A}\,,
\ee
with $$\lambda_{A}:=\sup\sigma(A)\,.
$$ 
Notice that (see \cite[relation (3.93)]{A})
\be\label{SC}
\mbox{Sin}_{A}(t)=\int_{0}^{t}\mbox{Cos}_{A}(s)\,ds\,.
\ee
If $\lambda_{A}=0$, then, by functional calculus,
$$
\mbox{Cos}_{A}(t)=\cos(t(-A)^{1/2})\,,\qquad \mbox{Sin}_{A}(t)=
(-A)^{-1/2}\sin(t(-A)^{1/2})\,.
$$ 
Given $\chi\in L^{1}(0,+\infty)$ 
and given $g\in L^{2}(\RE^{3})$,  let $u^{A}_{\chi g}$ be the solution of the wave equation with the source $\chi g$, i.e., 
\be\label{Cpb-inh}
\begin{cases}
\partial_{tt} u^{A}_{\chi g}(t)=A\,u^{A}_{\chi g}(t)+\chi(t)g\\
u^{A}_{\chi g}(0)=0\\
\partial_{t} u^{A}_{\chi g}(0)=0\,.
\end{cases}
\ee
By \cite[Proposition 3.1.16]{A} (see also \cite[Section II.4]{Fatto}), 
\be\label{Sin}
u^{A}_{\chi g}(t)=\int_{0}^{t}\mbox{Sin}_{A}(t-s)\chi(s)g\,ds
\,.
\ee
Let $\chi_{\tau}\in L^{1}(0,+\infty)$ be an approximation of Dirac's delta distribution at $t=0$, i.e.,
\be\label{chi}
\chi_{\tau}(t)\ge 0\,,\qquad \int_{0}^{+\infty}\!\!\!\chi_{\tau}(s)\,ds=1\,,\qquad 
\lim_{\tau\searrow 0}\int_{0}^{+\infty}\!\!\!s\chi_{\tau}(s)\,ds =0\,.
\ee
Two common choices are $\chi_{\tau}(t)=\frac1\tau\,1_{[0,\tau]}(t)$ and  
$\chi_{\tau}(t)=\frac1\tau\,e^{-t/\tau}$.\par
Let $u^{A}_{g}(t)$ be the solution of the homogenous  Cauchy problem
\be\label{Cpb-h}
\begin{cases}
\partial_{tt} u^{A}_{g}(t)=A\,u^{A}_{g}(t)\\
u^{A}_{g}(0)=0\\
\partial_{t} u^{A}_{g}(0)=g\,.
\end{cases}
\ee
By \eqref{sol}, \eqref{Sin}, \eqref{SC} and hypotheses \eqref{chi}, one gets
\begin{align*}
&\lim_{\tau\searrow 0}\big\|u^{A}_{g}(t)-u^{A}_{\chi_{\tau} g}(t)\big\|_{L^{2}}=
\lim_{\tau\searrow 0}\Big\|\,\mbox{Sin}_{A}(t)g-\int_{0}^{t}\mbox{Sin}_{A}(t-s)\chi_{\tau}(s)g\,ds\,\Big\|_{L^{2}}\\
\le&\lim_{\tau\searrow 0}\left(\,\int_{0}^{t}\bigg\|\bigg(\int_{t-s}^{t}\mbox{Cos}_{A}(r)\,dr\bigg)\chi_{\tau}(s)g\,\bigg\|_{L^{2}}\!ds+\int_{t}^{+\infty}\big\|\mbox{Sin}_{A}(t)\chi_{\tau}(s)g\,\big\|_{L^{2}}\,ds\right)\\
\le&\lim_{\tau\searrow 0}\left(\,\sup_{0\le r\le t}\|\mbox{Cos}_{A}(r)\|_{L^{2}\!,L^{2}}\|g\|_{L^{2}}\int_{0}^{t}s\chi_{\tau}(s)\,ds +\|\mbox{Sin}_{A}(t)\|_{L^{2}\!,L^{2}}\|g\|_{L^{2}}\int_{t}^{+\infty}\!\!\!\chi_{\tau}(s)\,ds\right)\\
\le &\,c\lim_{\tau\searrow 0}\int_{0}^{+\infty}\!\!\! s\chi_{\tau}(s)\,ds=0\,.
\end{align*}
Hence, the $u^{A}_{g}(t)$ solving \eqref{Cpb-h} can be interpreted  as the solution of the inhomogeneous Cauchy problem
\be\label{Cpb-d}
\begin{cases}
\partial_{tt} u^{A}_{g}(t)=A\,u^{A}_{g}(t)+\delta_{0}(t)g\\
u^{A}_{g}(0)=0\\
\partial_{t} u^{A}_{g}(0)=0\,.
\end{cases}
\ee 
\begin{remark}\label{WD} By \eqref{sin}, if $A_{n}$ converges to $A$ in strong resolvent sense as $n\nearrow+\infty$, then
$$
\lim_{n\nearrow+\infty}\int_{0}^{\infty}e^{-\sqrt\lambda\, t}u^{A_{n}}_{g}\,dt=
\int_{0}^{\infty}e^{-\sqrt\lambda\, t}u^{A}_{g}(t)\,dt\,.
$$
\end{remark}
\subsection{The data operator with point scatterers.} Let 
$u^{\alpha,Y}_{\chi_{\tau} g}$ and $u^{\varnothing}_{\chi_{\tau} g}$ denote the solutions of the Cauchy problems 
$$
\begin{cases}
\partial_{tt} u^{\alpha,Y}_{\chi_{\tau} g}(t,x)=\Delta_{\alpha,Y}\,u^{\alpha,Y}_{\chi_{\tau} g}(t,x)+\chi_{\tau}(t)g(x)\\
u^{\alpha,Y}_{\chi_{\tau} g}(0,x)=0\\
\partial_{t} u^{\alpha,Y}_{\chi_{\tau} g}(0,x)=0\,,
\end{cases}
\quad 
\begin{cases}
\partial_{tt} u^{\varnothing}_{\chi_{\tau} g}(t,x)=\Delta\,u^{\varnothing}_{\chi_{\tau} g}(t,x)+\chi_{\tau}(t)g(x)\\
u^{\varnothing}_{\chi_{\tau} g}(0,x)=0\\
\partial_{t} u^{\varnothing}_{\chi_{\tau} g}(0,x)=0\,.
\end{cases}
$$
With respect to the previous subsection, here we use the notations 
$$
u^{\alpha,Y}_{(...)}\equiv u^{\Delta_{\alpha,Y}}_{(...)}\,,\qquad 
u^{\varnothing}_{(...)}\equiv u^{\Delta}_{(...)}\,.
$$
In typical scattering experiments one measures the scattered wave 
\be\label{sw}
u^{\alpha,Y}_{\chi_{\tau} g}(t,x)-u^{\varnothing}_{\chi_{\tau} g}(t,x)
\ee
produced by a sharp pulse $\chi_{\tau}g$, $\tau\ll 1$. By the previous discussion leading to \eqref{Cpb-d} (equivalently to \eqref{Cpb-h}), in the ideal experiment in which the pulse is concentrated at $t=0$, \eqref{sw} is replaced by 
$$
u^{\alpha,Y}_{g}(t,x)-u^{\varnothing}_{g}(t,x)
\,,
$$
where $u^{\alpha,Y}_{g}$ and $u^{\varnothing}_{g}$ solve 
\be\label{CP2}
\begin{cases}
\partial_{tt} u^{\alpha,Y}_{ g}(t,x)=\Delta_{\alpha,Y}\,u^{\alpha,Y}_{g}(t,x)\\
u^{\alpha,Y}_{g}(0,x)=0\\
\partial_{t} u^{\alpha,Y}_{g}(0,x)=g(x)\,,
\end{cases}
\qquad 
\begin{cases}
\partial_{tt} u^{\varnothing}_{g}(t,x)=\Delta\,u^{\varnothing}_{g}(t,x)\\
u^{\varnothing}_{g}(0,x)=0\\
\partial_{t} u^{\varnothing}_{g}(0,x)=g(x)\,.
\end{cases}
\ee
\par 
Considering then an array of points $X_{N}=\{x_{1},\dots,x_{N}\}\subset\RE^{3}\backslash Y$, we now assume  $g=f_{\epsilon}$  supported in a $\epsilon$-neighborhood of $X_{N}$, where 
\be\label{FE}
f_{\epsilon}(x):=\sum_{k=1}^{N}f_{k}\,\varphi_{\epsilon}(x-x_{k})\,,\qquad f\equiv(f_{1},\dots f_{N})\,,
\ee
\be\label{ac}
\varphi_{\epsilon}(x)=\frac1{\epsilon^{3}}\,\varphi\left(\frac{x}{\epsilon}\right)\,,\quad
 \varphi\in{\mathscr C}^{\infty}_{\comp}(\RE^{3})\,,\quad \int_{\RE^{3}}\varphi(x)\,dx=1\,.
\ee
By \eqref{FE} and \eqref{ac}, $f_{\epsilon}$ converges, in distributional sense, to $\sum_{k=1}^{N}f_{k}\,\delta_{x_{k}}$ as $\epsilon\searrow 0$.
Let us introduce the operator 
$$
F^{N,\epsilon}_{\lambda}:\RE^{N}\to\RE^{N}\,,\qquad  (F^{N,\epsilon}_{\lambda}f)_{k}:=
\int_{0}^{\infty}e^{-\sqrt\lambda\,t}
(u^{\alpha,Y}_{f_{\epsilon}}(t,x_{k})-u^{\varnothing}_{f_{\epsilon}}(t,x_{k}))\,dt\,,\quad \lambda>\lambda_{\alpha,Y}\,.
$$
$F^{N,\epsilon}_{\lambda}$ corresponds to the measurements at time $t$ and at points $x_{1},\dots,x_{N}$ of the scattered waves produced by pulses supported at $t=0$ and in tiny (whenever $\epsilon\ll 1$) neighborhoods  of the same points $x_{1},\dots,x_{N}$ (detectors and emitters are at the same places). 
\begin{remark}
Since $f_{\epsilon}\in{\mathscr C}^{\infty}_{\comp}(\RE^{3})$ belongs to the form domain of $\Delta_{\alpha,Y}$ (that is to $\dom(Q_{\alpha,Y})$ as defined in the proof of Lemma \ref{form}), the solution $u^{\alpha,Y}_{f_{\epsilon}}(t,\cdot)$ entering in the definition of $F^{N,\epsilon}_{\lambda}$ is a strong one (see e.g. \cite[Chapter
2, section 7]{Gold}), i.e.,  $$u^{\alpha,Y}_{f_{\epsilon}}\in {\mathscr C}(\RE_{+},\dom(\Delta_{\alpha,Y}))\cap {\mathscr C^{1}}(\RE_{+},\dom(Q_{\alpha,Y}))\cap {\mathscr C^{2}}(\RE_{+},L^{2}(\RE^{3}))\,.
$$ 
Since $\dom(\Delta_{\alpha,Y})\subset {\mathscr C}(\RE^{3}\backslash Y)$, the evaluation at the point $x_{k}$ in $(F^{N,\epsilon}_{\lambda}f)_{k}$ is a legitimate operation.
\end{remark}
Next we show that $F^{N,\epsilon}_{\lambda}$ admits  a well defined limit as $\epsilon\searrow 0$, so that one is allowed to consider  the case in which the $N$ emitters and the $N$ detectors are both placed at the points $x_{1},\dots,x_{N}$.
\begin{lemma}\label{LF} Let $u^{\alpha,Y}_{f_{\epsilon}}$ and $u^{\varnothing}_{f_{\epsilon}}$ be the solutions of the Cauchy problems in \eqref{CP2}, where $g$ equals $f_{\epsilon}$ defined in \eqref{FE}. Then, the limits
$$
\lim_{\epsilon\searrow 0}\,\big(F^{N,\epsilon}_{\lambda}f\big)_{k}=\lim_{\epsilon\searrow 0}
\int_{0}^{\infty}e^{-\sqrt\lambda\,t}\big(u^{\alpha,Y}_{f_{\epsilon}}(t,x_{k})-u^{\varnothing}_{f_{\epsilon}}(t,x_{k})\big)\,dt\,,\qquad k=1,\dots, N\,,
$$
exist.
\end{lemma}
\begin{proof}
Since $ u^{\alpha,Y}_{f_{\epsilon}}$ and $u^{\varnothing}_{f_{\epsilon}}$ solve \eqref{CP2} with $g=f_{\epsilon}$, by \eqref{sol}, \eqref{sin} and the resolvent formula \eqref{krein}, one obtains
\begin{align*}
\lim_{\epsilon\searrow 0}\,(F^{N,\epsilon}_{\lambda}f)_{k}=&\lim_{\epsilon\searrow 0}
\int_{0}^{\infty}e^{-\sqrt\lambda\,t}\big(u^{\alpha,Y}_{f_{\epsilon}}(t,x_{k})-u^{\varnothing}_{f_{\epsilon}}(t,x_{k})\big)\,dt\nonumber\\
=&\lim_{\epsilon\searrow 0}\int_{0}^{\infty}e^{-\sqrt\lambda\,t}\big(\mbox{Sin}_{\Delta_{\alpha,Y}}(t)f_{\epsilon})(x_{k})-((-\Delta)^{-1/2}\sin(t(-\Delta)^{1/2})f_{\epsilon}\big)(x_{k})\,dt\nonumber\\
=&\lim_{\epsilon\searrow 0}\big((-\Delta_{\alpha,Y}+\lambda)^{-1}f_{\epsilon}-(-\Delta+\lambda)^{-1}f_{\epsilon}\big)(x_{k})\nonumber\\
=&\frac1{(4\pi)^{2}}\sum_{i,j=1}^{n}\Lambda_{\lambda}^{ij}\ \frac{e^{-\sqrt\lambda\,|x_{k}-y_{i}|}}{|x_{k}-y_{i}|}\,\left(\lim_{\epsilon\searrow 0}\int_{\RE^{3}}\frac{e^{-\sqrt\lambda\,|x-y_{j}|}}{|x-y_{j}|}\,f_{\epsilon}(x)\,dx\right)
\end{align*}
By \eqref{FE} and \eqref{ac},
\begin{align}\label{data1}
\lim_{\epsilon\searrow 0}\,(F^{N,\epsilon}_{\lambda}f)_{k}=&\lim_{\epsilon\searrow 0}
\int_{0}^{\infty}e^{-\sqrt\lambda\,t}\big(u^{\alpha,Y}_{f_{\epsilon}}(t,x_{k})-u^{\varnothing}_{f_{\epsilon}}(t,x_{k})\big)\,dt\nonumber\\
=&
\frac1{(4\pi)^{2}}\sum_{i,j=1}^{n}\Lambda_{\lambda}^{ij}\ \frac{e^{-\sqrt\lambda\,|x_{k}-y_{i}|}}{|x_{k}-y_{i}|}\,\sum_{\ell=1}^{N}\frac{e^{-\sqrt\lambda\,|x_{\ell}-y_{j}|}}{|x_{\ell}-y_{j}|}\,f_{\ell}\,.
\end{align}
\end{proof}
\subsection{A convenient representation of the scattered waves.}\label{MS} In order to implement numerical tests, it is useful to have at disposal an explicit formula for the difference of the solutions of the two Cauchy problems \eqref{CP2}, providing a convenient representation of the scattered waves entering in the definition of  $F^{N,\epsilon}_{\lambda}$.  By \cite[Theorem 3]{KP} (see also \cite[Theorem 3.1]{NP05}, and, for the case of a single point scatterer, the antecedent result in \cite[Theorem 3.2]{NP98}) $u^{\alpha,Y}_{f_{\epsilon}}(t)$ can be written in terms of $u^{\varnothing}_{f_{\epsilon}}(t)$ and of the solution of a system of inhomogeneous retarded
first-order differential equations.\par 
More precisely, if $q_{\epsilon}(t)\equiv(q_{\epsilon}^{1}(t),\dots, q^{n}_{\epsilon}(t))$, $t\ge 0$, denotes the unique solution of the Cauchy problem (here the dot in $\dot q_{\epsilon}^{j}(t)$ denotes the time-derivative and $H$ is Heaviside's function)
\be\label{ret}
\begin{cases}
\frac1{4\pi}\,\dot q_{\epsilon}^{j}(t)+
\alpha_{j}\,q_{\epsilon}^{j}(t)= u^{\varnothing}_{f_{\epsilon}}(t,y_{j})+
\sum_{i\not=j}\,\frac{H(t-|y_{i}-y_{j}|)}{4\pi\,|y_{i}-y_{j}|}\ q_{\epsilon}^{i}(t-|y_{i}-y_{j}|)
\,,\\
q_{\epsilon}^{j}(0)=0\,,\qquad j=1,\dots, n\,,
\end{cases}
\ee
then
\begin{equation}\label{cnv}
u^{\alpha,Y}_{f_{\epsilon}}(t,x)=u^{\varnothing}_{f_{\epsilon}}(t,x)+\sum_{j=1}^{n}\,\frac{H(t-|x-y_{j}|)}{4\pi\,|x-y_{j}|}\ q_{\epsilon}^{j}(t-|x-y_{j}|)\,, 
\end{equation}
i.e., $u^{\alpha,Y}_{f_{\epsilon}}(t)$ coincides with the solution $u_{\epsilon}(t)$, of the inhomogeneous (distributional) Cauchy problem
\be\label{distr}
\begin{cases}
\partial_{tt} u_{\epsilon}(t)=\Delta  u_{\epsilon}(t)+\sum_{j=1}^{n} q_{\epsilon}^{j}(t)\,\delta_{y_{j}}+\delta_{0}(t)f_{\epsilon}
\\
u_{\epsilon}(0)=0\\
\partial_{t} u_{\epsilon}(0)=0\,,
\end{cases}
\ee
where the $q_{\epsilon}^{j}(t)$'s solve \eqref{ret}. Notice  that the retarded terms in \eqref{ret} take into account the multiple scattering effects.\par
In conclusion,
$$
u^{\alpha,Y}_{f_{\epsilon}}(t,x_{k})-u^{\varnothing}_{f_{\epsilon}}(t,x_{k})=\sum_{j=1}^{n}\,\frac{H(t-|x_{k}-y_{j}|)}{4\pi\,|x_{k}-y_{j}|}\ q_{\epsilon}^{j}(t-|x_{k}-y_{j}|)\,,
$$
where the $q_{\epsilon}^{j}(t)$'s solve \eqref{ret}.
\section{Inverse wave scattering in the time domain.}
Let us fix $\lambda>\lambda_{\alpha,Y}$, a compact set $K\supset Y$ and a denumerable set $$
D=\{x_{k},\ k\in{\mathbb N}\}\subset K\backslash Y\,;
$$
$D$ represents the points where the emitters/detectors can be placed. We introduce the following hypothesis regarding $D\,$:
\be\label{(H)}
\text{the closure of $D$ contains a not void open set.}
\ee
For any integer $N>0$, we define the map 
\be\label{phi}
\phi_{\lambda}^{N}:K\backslash D\to \RE^{N}\,,\quad
\phi_{\lambda}^{N}(z)\equiv\begin{bmatrix}
\frac{e^{-\sqrt\lambda\,|x_{1}-z|}}{|x_{1}-z|}\\
\vdots\\
\frac{e^{-\sqrt\lambda\,|x_{N}-z|}}{|x_{N}-z|}
\end{bmatrix}
\ee
and the linear operator
\be\label{Phi}
\Phi^{N}_{\lambda}:\RE^{n}\to \RE^{N}\,,
\qquad\Phi_{\lambda}^{N}\equiv\begin{bmatrix}\frac{e^{-\sqrt\lambda\,|x_{1}-y_{1}|}}{|x_{1}-y_{1}|}&\dots&\frac{e^{-\sqrt\lambda\,|x_{1}-y_{n}|}}{|x_{1}-y_{n}|}\\
\vdots&{\ }&\vdots\\
\frac{e^{-\sqrt\lambda\,|x_{N}-y_{1}|}}{|x_{N}-y_{1}|}&\dots&\frac{e^{-\sqrt\lambda\,|x_{N}-y_{n}|}}{|x_{N}-y_{n}|}\end{bmatrix}.
\ee
\begin{lemma}\label{lemma} Let $\phi_{\lambda}^{N}$ and $\Phi^{N}_{\lambda}$ be the maps defined in \eqref{phi} and \eqref{Phi}. Then, there exists $N_{\circ}\ge 1$ such that, for any $N\ge N_{\circ}$,
$$
z\in Y\quad\iff\quad \phi_{\lambda}^{N}(z)\in\ran(\Phi^{N}_{\lambda})\,.
$$
\end{lemma}
\begin{proof} ($\Rightarrow$) If $z=y_{k}\in Y$, then $\phi_{\lambda}^{N}(z)=\Phi^{N}_{\lambda}\xi$, where $\xi\equiv(\xi_{1},\dots,\xi_{n})$, 
$\xi_{j}=\delta_{kj}$. \par 
($\Leftarrow$) Here we mimic the arguments provided in the proof of \cite[Theorem 4.1]{KG}. Suppose that the implication is false, i.e., 
$$
\text{$\forall N\ge 1$, $\exists z_{N}\in K\backslash Y$ such that $\phi^{N}_{\lambda}(z_{N})\in\ran(\Phi^{N}_{\lambda})$.}
$$ 
Since $\Phi^{N}_{\lambda}=\left[\phi^{N}_{\lambda}(y_{1}),\dots,\phi^{N}_{\lambda}(y_{n})\right]$, 
$$\phi^{N}_{\lambda}(z_{N})\in\ran(\Phi^{N}_{\lambda})\quad\iff\quad \phi^{N}_{\lambda}(z_{N})\in \text{span}\!\left\{\phi^{N}_{\lambda}(y_{1}),\dots,\phi^{N}_{\lambda}(y_{n})\right\}\,.
$$
an so there would exist sequences 
$$
\{N_{\ell}\}_{\ell=1}^{\infty}\subset{\mathbb N}\,,\ N_{\ell}\nearrow+\infty\,,\quad\{\xi_{\ell}\}_{\ell=1}^{\infty}\subset \RE^{n}\,,\ \{\eta_{\ell}\}_{\ell=1}^{\infty}\subset\RE\,,\ 
|\xi_{\ell}|^{2}+|\eta_{\ell}|^{2}=1\,,\quad \{z_{\ell}\}_{\ell=1}^{\infty}\subset K\backslash Y\,,
$$
such that 
\be\label{LD}
\forall \ell\ge 1\,,\quad\sum_{j=1}^{n}(\xi_{\ell})_{j}\,\phi^{N_{\ell}}_{\lambda}(y_{j})=\eta_{\ell}\,\phi^{N_{\ell}}_{\lambda}(z_{\ell})\,.
\ee
Therefore the analytic functions
$$
v_{\ell}:\RE^{3}\backslash (Y\cup\{z_{\ell}\})\to\RE\,,\quad v_{\ell}(x):=\eta_{\ell}\,\frac{e^{-\sqrt\lambda\,|x-z_{\ell}|}}{|x-z_{\ell}|}-\sum_{j=1}^{n}(\xi_{\ell})_{j}\,\frac{e^{-\sqrt\lambda\,|x-y_j|}}{|x-y_j|} 
$$
would vanish on the set $D$ and so, by our hypothesis \eqref{(H)}, they would be identically zero. Hence
$$
\forall \ell\ge 1\,,\ \forall x\in \RE^{3}\backslash (Y\cup\{z_{\ell}\})\,,\quad
\eta_{\ell}\,\frac{e^{-\sqrt\lambda\,|x-z_{\ell}|}}{|x-z_{\ell}|}=\sum_{j=1}^{n}(\xi_{\ell})_{j}\,\frac{e^{-\sqrt\lambda\,|x-y_j|}}{|x-y_j|} 
\,.
$$
Since the sequences $\{\xi_{\ell}\}_{\ell=1}^{\infty}$, $\{\eta_{\ell}\}_{\ell=1}^{\infty}$ and 
$\{z_{\ell}\}_{\ell=1}^{\infty}$ are bounded, one has, eventually considering subsequences, 
$\xi_{\ell}\to\xi\in\RE^{n}$,   $\eta_{\ell}\to\eta\in\RE$, $z_{\ell}\to z\in K$ as $\ell\nearrow+\infty$ and so one would get 
\be\label{eq}
 \forall x\in \RE^{3}\backslash (Y\cup\{z\})\,,\quad
\eta\,\frac{e^{-\sqrt\lambda\,|x-z|}}{|x-z|}=\sum_{j=1}^{n}\xi_{j}\,\frac{e^{-\sqrt\lambda\,|x-y_j|}}{|x-y_j|} 
\,.
\ee
Let us now show that this is impossible, by considering separately the cases $z\in K\backslash Y$ and 
$z=y_{i}\in Y$. \par 
If $z\in K\backslash Y$ then, by considering the limit $x\to y_{k}$
in \eqref{eq}, one would get $\xi_{k}=0$ for any $k$ and hence $\eta=0$; this  is impossible, since $|\xi|^{2}+|\eta|^{2}=1$. \par
If $z=y_{i}\in Y$ then, by considering the limit $x\to y_{i}$
in \eqref{eq}, one would get $\xi_{i}=\eta$ and therefore
\be\label{eqk}
 \forall x\in \RE^{3}\backslash (Y\cup\{z\})\,,\qquad \sum_{j\not=i}\xi_{j}\,\frac{e^{-\sqrt\lambda\,|x-y_j|}}{|x-y_j|} =0
\,.
\ee
This would give $\xi_{j}=0$ for any $j\not=i$ and hence $|\eta|=\frac1{\sqrt 2}$.
Now, let us re-write \eqref{LD} as
\begin{align}\label{LD1}
\forall k\ge 1\,,\quad&\sum_{j\not=i}\frac{(\xi_{\ell})_{j}}{\varepsilon_{\ell}}\,
\frac{e^{-\sqrt\lambda\,|x_{k}-y_{j}|}}{|x_{k}-y_{j}|}
=\frac{\eta_{\ell}-(\xi_{\ell})_{k}}{\varepsilon_{\ell}}\,\frac{e^{-\sqrt\lambda\,|x_{k}-y_{i}|}}{|x_{k}-y_{i}|}\nonumber\\
&+\frac{\eta_{\ell}}{\varepsilon_{\ell}}\left( \frac{e^{-\sqrt\lambda\,|x_{k}-z_{\ell}|}}{|x_{k}-z_{\ell}|}-
\frac{e^{-\sqrt\lambda\,|x_{k}-y_{i}|}}{|x_{k}-y_{i}|}\right)\,,
\end{align}
where 
$$
\varepsilon_{\ell}:=\bigg(|\eta_{\ell}-(\xi_{\ell})_{i}|^{2}+\sum_{j\not=i}|(\xi_{\ell})_{j}|^{2}+|z_{\ell}-y_{i}|^{2}\bigg)^{1/2}\,.
$$
By
\begin{align*}
&\frac{e^{-\sqrt\lambda\,|x_{k}-z_{\ell}|}}{|x_{k}-z_{\ell}|}-
\frac{e^{-\sqrt\lambda\,|x_{k}-y_{i}|}}{|x_{k}-y_{i}|}\\
=&\left(\sqrt\lambda+\frac1{|x_{k}-y_{i}|}\right)
\frac{e^{-\sqrt\lambda\,|x_{k}-y_{i}|}}{|x_{k}-y_{i}|}\,
\frac{x_{k}-y_{i}}{|x_{k}-y_{i}|}\cdot(z_{\ell}-y_{i})+o(|y_{i}-z_{\ell}|)\,,
\end{align*}
and by (eventually considering subsequences)
$$
\frac{(\xi_{\ell})_{j}}{\varepsilon_{\ell}}\to \tilde\xi_{j}\,,\ j\not=i\,,\quad 
\frac{(\xi_{\ell})_{i}-\eta_{\ell}}{\varepsilon_{\ell}}\to \tilde\xi_{i}\,,\quad
\frac{z_{\ell}-y_{k}}{\varepsilon_{\ell}}\to \tilde z_{k}\,,\quad |\tilde\xi|^{2}+|\tilde z_{k}|^{2}=1\,,
$$
as $\ell\nearrow+\infty$, one would get, considering the limit $\ell\nearrow+\infty$ in the relation \eqref{LD1},
$$
\forall k\ge 1\,,\quad\sum_{j=1}^{n}\tilde\xi_{j}\,
\frac{e^{-\sqrt\lambda\,|x_{k}-y_{j}|}}{|x_{k}-y_{j}|}
=\eta\left(\sqrt\lambda+\frac1{|x_{k}-y_{i}|}\right)
\frac{e^{-\sqrt\lambda\,|x_{k}-y_{i}|}}{|x_{k}-y_{i}|}\,\frac{x_{k}-y_{i}}{|x_{k}-y_{i}|}\cdot \tilde z_{k}
$$
Again taking into account hypothesis \eqref{(H)} on the set $D$, one would obtain 
\be\label{wo}
\forall x\in \RE^{3}\backslash Y,\quad\sum_{j=1}^{n}\tilde\xi_{j}\ 
\frac{e^{-\sqrt\lambda\,|x-y_{j}|}}{|x-y_{j}|}
=\eta\left(\sqrt\lambda+\frac1{|x-y_{i}|}\right)
\frac{e^{-\sqrt\lambda\,|x-y_{i}|}}{|x-y_{i}|}\,\frac{x-y_{i}}{|x-y_{i}|}
\cdot \tilde z_{k}\,.
\ee
Considering the limits $x\to y_{j}$, $j\not=i$ in \eqref{wo}, one would get $\tilde\xi_{j}=0$ for any $j\not=i$ and so \eqref{wo} would reduce to 
\be\label{wo1}
\forall x\in \RE^{3}\backslash Y\,,\qquad\tilde\xi_{i}
=\eta\left(\sqrt\lambda+\frac1{|x-y_{i}|}\right)
\frac{x-y_{i}}{|x-y_{i}|}
\cdot \tilde z_{k}\,,
\ee
Considering the limit $x\to y_{i}$ in \eqref{wo1}, one would get $\tilde z_{k}=0$ and hence $\tilde\xi_{i}=0$. This is impossible, since  $|\tilde\xi_{i}|^{2}+|\tilde z_{k}|^{2}=1$. 
\end{proof}
\begin{theorem}\label{thm} Let  $D=\{x_{k},\ k\in{\mathbb N}\}\subset K\backslash Y$ satisfy hypothesis \eqref{(H)} and let $\lambda>\lambda_{\alpha,Y}$. Then there exists $N_{\circ}\ge 1$ such that for any $N\ge N_{\circ}$ the data operator corresponding to $X_{N}:=\{x_{k}\in D,\ k\le N\}$ defined by
$$
F^{N}_{\lambda}:\RE^{N}\to\RE^{N}\,,\quad (F^{N}_{\lambda}f)_{k}:=
\lim_{\epsilon\searrow 0}
\int_{0}^{\infty}e^{-\sqrt\lambda\,t}\big(u^{\alpha,Y}_{f_{\epsilon}}(t,x_{k})-u^{\varnothing}_{f_{\epsilon}}(t,x_{k})\big)\,dt
$$ 
determines $Y$ according to
$$z\in Y\quad\iff\quad \phi_{\lambda}^{N}(z)\in \ker(F^{N}_{\lambda})^{\perp}\,.
$$
Equivalently, denoting by $P^{N}_{\lambda}$ the orthogonal projector onto $\ker(F_{\lambda}^{N})$, one has
$$Y=\{\,\text{peak points of the function}\ 
z\mapsto|P^{N}_{\lambda}\phi_{\lambda}^{N}(z)|^{-1}\,\}\,.
$$
\end{theorem}
\begin{proof} By \eqref{data1}, one has
\be\label{data2}
F^{N}_{\lambda}=(4\pi)^{-2}\,\Phi^{N}_{\lambda} \Lambda_{\lambda}(\Phi^{N}_{\lambda})^{*} \,.
\ee
Since $M_{\lambda}$ is positive-definite by Lemma \eqref{form}, $\Lambda_{\lambda}=M_{\lambda}^{-1}$ is positive-definite as well and so it has a nonsingular square root. Hence
$$
F^{N}_{\lambda}=(4\pi)^{-2}\,\Phi^{N}_{\lambda} \Lambda^{1/2}_{\lambda}\big(\Phi^{N}_{\lambda} \Lambda^{1/2}_{\lambda}\big)^{*}
$$
and $$\ker(F^{N}_{\lambda})^{\perp}=\ran\big((F^{N}_{\lambda})^{*}\big)=\ran\big(\Phi^{N}_{\lambda} \Lambda^{1/2}_{\lambda}\big(\Phi^{N}_{\lambda} \Lambda^{1/2}_{\lambda}\big)^{*}\big)=\ran\big(\Phi^{N}_{\lambda} \Lambda^{1/2}_{\lambda}\big)=\ran(\Phi^{N}_{\lambda})\,.$$
The proof is then concluded by Lemma \ref{lemma}.
\end{proof}
\begin{remark} The arguments of proof of Lemma \ref{lemma}, being not constructive, do not provide any lower bound on the number  $N_{\circ}$ of emitters/detectors needed to locate the $n$ point in the array $Y$. Since the $N\times N$ data matrix $F^{N}_{\lambda}$ is symmetric, by elementary dimensional considerations, in order to determine the unknown $3n$ coordinates, one at least would need $N_{\circ}$ such that $(N_{\circ}^{2}+N_{\circ})/2\ge 3n$. This suggests the rough estimate
$$
N_{\circ}>\left\lfloor\frac{5\sqrt n-1}{2}\right\rfloor\,, 
$$  
where $\lfloor\cdot\rfloor$ denotes the integer part. At the best of our knowledge, 
a precise estimate of $N_{\circ}$ is a relevant open question which may require different techniques from the ones implemented here.  
\end{remark}
\begin{remark} Suppose that the locations $y_{1}, \dots,y_{n}$ of the scatterers have been determined.
Let us now work with $N=n$ sources/detectors. Thus the $n\times n$ matrix $\Phi_{\lambda}^{n}$ is known.\par 
Since the map $(0,+\infty)\ni\lambda\mapsto\det(\Phi_{\lambda}^{n})$ is analytic, it
has isolated zeroes; let us choose $\lambda >\lambda_{\alpha,Y}$ such that $\det(\Phi_{\lambda }^{n})\not=0$. Then, by \eqref{data1}, equivalently, by \eqref{data2},
$$
\Lambda_{\lambda }=M_{\lambda }^{-1}=(4\pi)^{2}\,(\Phi^{n}_{\lambda })^{-1}F^{n}_{\lambda }((\Phi^{n}_{\lambda })^{*})^{-1}\,.
$$
and
$$
\det(F^{n}_{\lambda })=(4\pi)^{-2n}(\det(\Phi^{n}_{\lambda }))^{2}\det(\Lambda_{\lambda })\not=0\,.
$$
Therefore, using \eqref{matrix M}, one can recover the coefficients $\alpha_{1},\dots,\alpha_{n}$ from the knowledge of $F_{\lambda}^{n}$ by the relations
$$
\alpha_{i}=\frac1{(4\pi)^{2}}\,\left((\Phi^{n}_{\lambda })^{*}(F^{n}_{\lambda })^{-1}\Phi^{n}_{\lambda} \right)_{ii}-\frac{\sqrt{\lambda }}{4\pi}\,.
$$
\end{remark}
\vskip20pt\noindent
{\bf Acknowledgements.} The authors thank Mourad Sini, Guanghui Hu and Ibtissem Ben A\"itcha for 
fruitful discussions. During the preparation of this work, A.M. was a member of the CNRS working group at WPI in Wien; he thanks the Pauli Institute for the hospitality and the support.  We also gratefully acknowledge one anonymous referee for the stimulating remarks.

\end{document}